\documentclass[12pt]{article}
\usepackage{amsmath}
\usepackage{amssymb,amsthm}
\usepackage{bbold}
\usepackage{color}
\newtheorem{Thm}{Theorem}[section]
\newtheorem{theorem}[Thm]{Theorem}
\newtheorem{proposition}[Thm]{Proposition}

\newtheorem{remark}{Remark}[section]

\newtheorem{definition}[Thm]{Definition}

\title{Unstable states in a model of nonrelativistic quantum electrodynamics: corrections to the Lorentzian distribution}

\author{Walter F. Wreszinski\footnote{wreszins@gmail.com, 
Instituto de Fisica, Universidade de S\~ao Paulo (USP), Brazil}}        

\begin{document}

\maketitle

\begin{abstract}
We revisit the Lee-Friedrichs model as a model of atomic resonances in the hydrogen atom,
using the dipole-moment matrix-element functions which have been exactly computed by
Nussenzveig. The Hamiltonian $H$ of the model is positive and has absolutely continuous 
spectrum. Although the return probability amplitude $R_{\Psi}(t) = (\Psi, \exp(-iHt) \Psi)$ of the initial 
state $\Psi$, taken as the so-called Weisskopf-Wigner (W.W.) state, cannot be computed exactly, we show that it 
equals the sum of an exponentially decaying term and a universal correction $O(\beta^{2}\frac{1}{t})$, 
for large positive times $t$ and small coupling constants $\beta$, improving on some results
of \cite{King}. The remaining, non-universal, part of the correction is also shown to be of the same 
qualitative type. The method consists in approximating the matrix element of the resolvent operator
operator in the W.W. state by a Lorentzian distribution. No use is made of complex energies associated 
to analytic continuations of the resolvent operator to ''unphysical'' Riemann sheets. Other new results 
are presented, in particular a physical interpretation of the corrections, and the characterization
of the so-called sojourn time $\tau_{H}(\Psi)= \int_{0}^{\infty} |R_{\Psi}(t)|^{2} dt$ as the average lifetime
of the decaying state, a standard quantity in (quantum) probability.   
\end{abstract}

\section{Introduction, motivation and synopsis. The model}

\subsection{Introduction and motivation}
\label{sec:1.1}

The problem of unstable states in quantum (field) theory has its origin in Gamow's early 
treatment of alpha decay (\cite{Gamow}, see also \cite{BrHa}). Its crucial importance to physics is due to two related 
facts: all atomic states - except for the ground state - are resonances, and, in elementary
particle physics, all but the lightest particles are unstable. In the former case, we have 
to do with a bound state problem of an atom in the presence of the electromagnetic field,
which is our subject in the present paper.

The first treatment of unstable (decaying) states of atoms in interaction with the
electromagnetic field was proposed by Weisskopf in his thesis, of which a lively account is
given in \cite{Weisskopf}. The ensuing paper by Weisskopf and Wigner \cite{WW} is the first paper
where a divergent integral appeared! The assumptions made by Weisskopf and Wigner were carefully
analysed and criticized by Davidovich and Nussenzveig (\cite{DaNu}, see also 
Davidovich's Ph.D. thesis \cite{Dav}). A review of their work, with
several improvements, was published by Nussenzveig in 1984 \cite{Nussenzveig}. We refer
to \cite{DaNu} for further references on the previous literature on the subject.

Davidovich and Nussenzveig were primarily concerned with providing a theory of natural line shape
of certain atomic levels, e.g., those concerned by the Lamb shift in the hydrogen atom \cite{JJS}.
Their approach may be summarized as follows: to identify, initially, in the full
Hamiltonian of interaction between the atoms and the electromagnetic field, a model for decaying states
which incorporates as many realistic features as possible, while remaining exactly soluble. The omitted
terms from the full Hamiltonian would then be dealt with by a suitable perturbation theory. When
specialized to $N=2$ atomic levels, their model coincides with a sector of the spin-Boson model in the 
rotating-wave approximation, whose spectrum was determined
by Friedrichs \cite{Friedrichs}, and is therefore known as the Friedrichs model (see also \cite{Howland}).
This model is also well-known in quantum field theory and particle physics as the Lee model \cite{Lee},
but here we shall revisit it as a model of atomic resonances in the hydrogen atom, 
using the dipole-moment matrix-element functions which have been exactly computed by Nussenzveig \cite{Nussenzveig}. 

In this paper we intend to clarify several points in their discussion, 
partly in view of a rigorous  result due to Christopher King \cite{King}, who revisited this model
in 1991. It is, however, important to mention that the problem of atomic resonances in nonrelativistic 
quantum electrodynamics has been treated at great length in an important series of papers by Bach, 
Froehlich and Sigal (see \cite{BaFroSi} and references given there, and \cite{GuSi} for a textbook account). 
They introduce, however, the electromagnetic vector potential field with an ultraviolet cutoff. 
Our model, in spite of several rather drastic approximations, has
no ultraviolet cutoff. In addition, as in \cite{King}, we do not adopt their concept of resonance, 
related to complex energies. The use of complex energies and frequencies, which is not a priori physically
motivated, leads to pathologies, such as the well-known ''exponential catastrophe'' in both classical and quantum
physics (see \cite{Nussenzveig} and section 2.1, and it seems therefore conceptually of great advantage to
avoid them, as we do in this paper. 

In section 1.2 we introduce the model and recall two well-known results, Theorems 1.1 and 1.2. Define the 
so-called return probability of the decaying state $\Psi$:
$$
|R_{\Psi}(t)|^{2} = |(\Psi, \exp(-iHt) \Psi)|^{2}
$$
where $\Psi$ is a specific (normalized to one) vector in the Fock space of atoms and field, 
and $(\cdot,\cdot)$ denotes the scalar product in this space. The corresponding amplitude $R_{\Psi}(t)$ 
is a basic quantity. The spectral measure of $H$ is absolutely continuous (see, e.g., \cite{BB} or \cite{MWB}),
and thus 
\begin{equation}
\label{(0.1)}
R_{\Psi}(t) = \int_{0}^{\infty} d\lambda \exp(-i \lambda t) g_{\Psi}(\lambda)
\end{equation}
for some locally (Lebesgue) integrable function $g_{\Psi}$. 

Theorem 1.2 is a well-known result, relating positivity of the
Hamiltonian and the rate of decay: the return probability amplitude cannot be a pure exponential, but must
be corrected by a term $c(t)= \epsilon(t)$, which we define in this paper as meaning: $c(t) \to 0$ as 
$t \to \infty$.

The ensuing section 2 is divided into two parts. In subsection 2.1, we briefly describe the method of ''decay without
analyticity'', which we follow in this paper, and was initiated by King \cite{King}, as well as briefly discuss the
associated ``time-arrow problem''. 

Subsection 2.2 is devoted to the proof of our main theorem, Theorem 2.1,
which states that the correction $c(t)$ is $O(\beta^{2}\frac{1}{t})$ for sufficiently small $\beta$
and large $t$. King \cite{King} made only minimal assumptions on the dipole-moment matrix element functions
and obtained only the $O(\beta^{2})$ part, although the Riemann-Lebesgue lemma implies that it is $\epsilon(t)$
(Theorem 1.2). We use the exact dipole-moment functions for hydrogen, which have been calculated by Nussenzveig 
in terms of hypergeometric functions \cite{Nussenzveig}. Our method of proof follows King \cite{King} 
and consists of comparing $g_{\Psi}(\lambda)$ in \eqref{(0.1)} with the Lorentzian or Breit-Wigner function

\begin{equation}
g_{\Psi}^{L}(\lambda) = \frac{\Gamma}{2\pi[(\lambda-\lambda_{0})^{2}+\frac{\Gamma^{2}}{4}]}
\label{(0.2)}
\end{equation}
If we insert \eqref{(0.2)} into \eqref{(0.1)} and replace the integral from zero to infinity by one from $-\infty$ to $\infty$, we
obtain
\begin{equation}
\label{(0.3)}
R_{\Psi}^{u}(t) = \exp(-i\lambda_{0}t - \frac{\Gamma t}{2}) 
\end{equation}
where the superscript $u$ stands for ''unbounded'', i.e., \eqref{(0.3)} corresponds to a non-semibounded Hamiltonian,
for which the spectrum extends to $-\infty$. \eqref{(0.3)} results from writing \eqref{(0.2)} as a sum of two pole contributions,
and further applying Cauchy's theorem along a contour  along the real line , closed by a large semi-circle
in the lower half plane, the latter's contribution vanishing if $t>0$. This is done by King \cite{King}, who proceeds from
this point to estimate the remainder. We use \eqref{(0.1)} directly, with the splitting 
$g_{\Psi} = g_{\Psi}^{L} + (g_{\Psi}-g_{\Psi}^{L})$: the integral \eqref{(0.1)} with $g_{\Psi}$ replaced by $g_{\Psi}^{L}$ 
is evaluated along a contour following the positive real line, a quarter circle at infinity in the lower half plane and coming back
along the negative imaginary axis. The latter's contribution yields a correction  $c(t)= O(\frac{1}{t})$ to the residue at
the pole (which coincides with the r.h.s. of \eqref{(0.3)}). This correction is universal and improves the results
of \cite{King}. The (non-universal) contribution of the remainder $g_{\Psi}-g_{\Psi}^{L}$ is shown to yield
a correction of the same type. This is the content of theorem 2.1, some details of which are left to appendix A.   

Finally, we make in section 3 an application of a time-energy uncertainty theorem (Theorem 3.17 of \cite{MWB}) 
to the present model, in order to find a 
lower bound to the energy fluctuation in the state $\Psi$ (Theorem 3.2). The significance of this theorem is better
appreciated by observing that this fluctuation equals
\begin{equation}
\label{(0.4)}
\int_{0}^{\infty} d\lambda \lambda^{2} g_{\Psi}(\lambda)^{2}
\end{equation}
but the same quantity, evaluated for the Lorentzian $g_{\Psi}^{L}$, is infinite.
In the process, it is also suggested that the time of sojourn
$\tau_{H}(\Psi)= \int_{0}^{\infty} |R_{\Psi}(t)|^{2} dt$ is the most natural quantity to consider
in connection with the decay of unstable atoms or particles: it is proved to coincide
with the the average lifetime of the decaying state, a standard quantity in quantum probability.

As in \cite{King}, no use is made of complex energies associated to analytic continuations
of the resolvent operator to ''unphysical'' Riemann sheets. In this paper, we are not concerned with thermal states. 
  
\subsection{The model}
\label{sec:1.2}

As mentioned in the previous section, in our account, we shall consider a prototypical model for the Lyman 
$\alpha$ transition in hydrogen: this will imply no qualitative restriction regarding the final results. We follow 
\cite{Nussenzveig} and choose his units 
$\hbar = c =1$; this still allows to set a unit of length, which is chosen as the Bohr radius
$a_{B}=(me^{2})^{-1}=(\alpha m)^{-1} = 1$, from which
        \begin{equation}
                \beta = \frac{e}{m} = \alpha^{3/2}
                \label{(1.1)}
        \end{equation}   
with
        \begin{equation}
                e = \alpha^{1/2}
                \label{(1.2)}
        \end{equation}
Above, $e,m$ denote charge and mass of the electron, and $\alpha$ the fine-structure constant, 
approximately equal to $\frac{1}{137}$. The ground state energy is
        \begin{equation}
                E_{01} = -\frac{\alpha}{2}
                \label{(1.3)}
        \end{equation}
and the resonant level (e.g., one of the two Lamb-shifted levels, degenerate in the Dirac theory \cite{JJS})
will have the energy $E_{0r}$; we denote
        \begin{equation}
                E_{0} = E_{0r}-E_{01}
                \label{(1.4)}
        \end{equation}

The model considered in (\cite{Nussenzveig}, \cite{DaNu}, \cite{Dav}), when specialized to $N=2$ atomic levels,
may be written
       \begin{equation}
               H = H_{0} + H_{I}
               \label{(1.5.1)}
       \end{equation}
with
       \begin{equation}
              H_{0} = E_{0} \frac{\mathbf{1}+\sigma_{z}}{2}\otimes \mathbf{1} + \mathbf{1}\otimes\int d^{3}k |k|a^{\dag}(k)a(k)
              \label{(1.5.2)}
       \end{equation}
and
       \begin{equation}
              H_{I} = \beta[\sigma_{-}\otimes a^{\dag}(g)+\sigma_{+}\otimes a(g)]
              \label{(1.5.3)}
       \end{equation}
The operators act on the Hilbert space
       \begin{equation}
             {\cal H} \equiv {\cal C}^{2} \otimes {\cal F}
             \label{(1.6)}
       \end{equation}
where ${\cal F}$ denotes symmetric (Boson) Fock space on $L^{2}(\mathbf{R}^{3})$ (see, e.g., \cite{MaRo}), which
describes the photons. We shall denote by $(\cdot,\cdot)$ the scalar product in ${\cal H}$. 
Formally, $a(g) \equiv \int d^{3}k g(k) a(k)$, and $k$ denotes a three-dimensional vector.
The $\dag$ denotes adjoint, $\sigma_{\pm}=\frac{\sigma_{x}\pm \sigma_{y}}{2}$, and $\sigma_{x,y,z}$ are the usual
Pauli matrices. The operator
       \begin{equation}
            N = \frac{\mathbf{1}+\sigma_{z}}{2}\otimes \mathbf{1}+ \mathbf{1} \otimes \int d^{3}k a^{\dag}(k)a(k)
            \label{(1.9)}
       \end{equation}
commutes with $H$. We write
       \begin{equation}
            N = \sum_{l=0}^{\infty} l P_{l}
            \label{(1.10)}
       \end{equation}
and introduce the notation
       \begin{equation}
           H_{l} \equiv P_{l} H P_{l}
           \label{(1.11)}
       \end{equation}
$H_{l}$ is the restriction of $H$ to the subspace $P_{l} {\cal H}$. The subspace $P_{0}{\cal H}$ is one-dimensional and consists
of the ground state vector
       \begin{equation}
           \Phi_{0} \equiv |-) \otimes |\Omega)
           \label{(1.12)}
       \end{equation}
with energy zero, where
       \begin{equation}
           \sigma_{z} |\pm) = \pm |\pm)
           \label{(1.13)}
       \end{equation}
denote the upper $|+)$ and lower $|-)$ atomic levels, and $|\Omega)$ denotes the zero-photon state in ${\cal F}$. Note that
$\Phi_{0}$ is also eigenstate of the free Hamiltonian $H_{0}$, with energy zero, and we say therefore that the model
has a persistent zero particle state. Thus, by a theorem in (\cite{Wight}, pg. 250) - which is logically independent from Haag's
theorem (\cite{Wight}, pg. 249), the model is well-defined in Fock space, and ${\cal H}$, defined by \eqref{(1.6)}, is, indeed,
the adequate Hilbert space.

We shall confine ourselves to the subspace $P_{1}{\cal H}$. Let
        \begin{equation}
           \Phi_{1} \equiv |+) \otimes |\Omega)
           \label{(1.14)}
        \end{equation}
be the so-called \emph{Weisskopf-Wigner state}, and
        \begin{equation}
           \Phi_{2}(h) \equiv |-) \otimes a^{\dag}(h)|\Omega) \mbox{ with } h \in L^{2}(\mathbf{R}^{3})
           \label{(1.15)}
        \end{equation}
The subspace $P_{1}{\cal H}$ consists of linear combinations
        \begin{equation}
            \Phi_{a,b} \equiv a \Phi_{1} + b \Phi_{2}(h)
            \label{(1.16)}
        \end{equation}
where $a,b$ are complex coefficients. This is the famous Friedrichs model \cite{Friedrichs}.

$E_{0}$ is given by \eqref{(1.4)} in the concrete case of the Lyman $\alpha$ transition, and
       \begin{equation}
            g(k)= g(|k|)= \sqrt{|k|}f(|k|)
            \label{(1.7)}
       \end{equation}
where
      \begin{equation}
            f(k) = (|k|^{2}+a^{2})^{-2}
            \label{(1.8.1)}
      \end{equation}
with
      \begin{equation}
            a= \frac{3}{2}
            \label{(1.8.2)}
      \end{equation}
with the choice of units \eqref{(1.1)}, \eqref{(1.2)}: the above functions $f$ are special dipole-moment matrix-element
functions for hydrogen, which may be computed explicitly in terms of hypergeometric functions (\cite{Nussenzveig}, (8.21)).
As mentioned, we take the above example as a prototype: consideration of the other cases in \cite{Nussenzveig} bring no
qualitative alterations in the forthcoming results.
Thus, $P_{1}{\cal H}$ becomes isomorphic to the space
        \begin{equation}
            {\cal H}_{1} \equiv {\cal C} \oplus L^{2}(0,\infty)
            \label{(1.17.2)}
        \end{equation}
with $H_{1} \equiv P_{1} H P_{1}$ is isomorphic to $H_{1}$ (using the same symbol) given by
        \begin{equation}  
            H_{1}=
            \left[\begin{array}{cc}
              E_{0} & \beta\langle g,. \rangle \\
            \beta g & |k| \\      
            \end{array}\right]
            \label{(1.18)}
        \end{equation}
where $g$ is given by \eqref{(1.7)}, \eqref{(1.8.1)}. The scalar product on $L^{2}(0,\infty)$ is denoted $\langle.,.\rangle$.
The following theorem follows from \cite{King} or (\cite{Howland}, Proposition 1, pg. 417):

\begin{theorem}
\label{th:1.1}

For the model \eqref{(1.17.2)}, \eqref{(1.18)}, let
        \begin{equation}
             E_{0} > \beta^{2} \int_{0}^{\infty}dk g(k)^{2}
             \label{(1.20)}
        \end{equation}
Then:
\begin{itemize}
\item [$a.)$] $H_{1}$ has spectrum 
         \begin{equation}
             \sigma(H_{1}) = [0,\infty)
             \label{(1.21.1)}
         \end{equation}
which is purely absolutely continuous. Furthermore, for all $z \in \mathbf{C}$ not in the positive
real axis:
\item [$b.)$] \begin{align}
              \label{(1.21.2)}
              \begin{split}
              & r_{\Phi_{1}}(z) \equiv (\Phi_{1}, (H_{1}-z)^{-1} \Phi_{1})\\
              & = (E_{0}-z-\beta^{2}\int_{0}^{\infty} dk\frac{g(k)^{2}}{k-z})^{-1}
              \end{split}
              \end{align}
\end{itemize}
\end{theorem}

The reason why, in \eqref{(1.21.2)}, the integral is not the three dimensional integral over the momentum variable,
but just a radial integral without the $k^{2}$ factor, will now be explained. Denote by $(.,.)$ the scalar product in ${\cal H}_{1}$. 
We have to do with the integral (see, e.g., (11) in \cite{King})
$$
(W,(|k|-z)^{-1} W) = \int d^{3}k |W(\vec{k})|^{2} (|\vec{k}|-z)^{-1}
$$
where $W(\vec{k})= \frac{1}{\sqrt(\omega(\vec{k}))}f(\vec{k})$, and
$$
f(\vec{k}) = \int d^{3}x \bar{\Psi}_{1}(\vec{x}) (e_{\vec{k}}.\vec{p})\Psi_{2}(\vec{x})\exp(i\vec{k}.\vec{x})
$$
where $\Psi_{1}$ and $\Psi_{2}$ are the wave-functions of the corresponding levels of hydrogen, $e_{\vec{k}}$
is a polarization vector, $\vec{p}$ the momentum operator and $\omega(\vec{k})= |\vec{k}|$ the photon energy
(see, e.g., \cite{JJS}, Chapter 2). Thus,
\begin{eqnarray*}
\int d^{3}k |W(\vec{k})|^{2} (|\vec{k}|-z)^{-1} = 4\pi \int_{0}^{\infty} dk k^{2} \frac{1}{k}|f(k)|^{2}(k-z)^{-1}=\\
= \int_{0}^{\infty} (\sqrt(k)f(k))^{2} (k-z)^{-1}
\end{eqnarray*}
going back to the notation $k \equiv |\vec{k}|$.
We have absorbed in the quantity $\beta^{2}$ in \eqref{(1.18)} the factor $4 \pi$ coming from 
integration over the solid angle.
Given the spectral family $\{E(\lambda)\}_{\lambda \in [0,\infty)}$ associated to $H_{1}$
(see, e.g., \cite{BB}), statement b.) of theorem ~\ref{th:1.1} means that the Stieltjes measure
(for the definition, see, e.g., \cite{Sewell1}, pg. 41):

          \begin{equation}
               d\mu_{\Phi_{1}}(\lambda) = (\Phi_{1}, E(\lambda) \Phi_{1}) = \int_{0}^{\lambda} g_{\Phi_{1}}(u) du
               \label{(1.23.1)}
          \end{equation}
where
          \begin{equation}
               g_{\Phi_{1}}(u) = \frac{d\mu_{\Phi_{1}}(u)}{du}
               \label{(1.23.2)}
          \end{equation}
exists almost everywhere (a.e.) in $u$ and defines a (locally) $L^{1}$ function. 
By a.) of theorem ~\ref{th:1.1} we
may express $g_{\Phi_{1}}$ in terms of $r_{\Phi_{1}}(z)$ by (\cite{RSI}, \cite{Jak}):
          \begin{equation}
               g_{\Phi_{1}}(\lambda) = \frac{1}{2\pi i}\lim_{\epsilon \to 0} [r_{\Phi_{1}}(\lambda+i\epsilon)-r_{\Phi_{1}}(\lambda-i\epsilon)]
               \label{(1.24)}
          \end{equation}
In spite of the exact result b.) of theorem ~\ref{th:1.1}, the time evolution of the initial state $\Phi_{1}$ is not 
explicitly known - a symptom of the complexity of the time evolution of quantum systems 
even in the simplest situations, and one must rely on suitable estimates.

We define the \emph{return probability amplitude} of the vector $\Phi_{1}$
\begin{equation}
\label{(2.9)}
R_{\Phi_{1}}(t) \equiv (\Phi_{1}, \exp(-iH_{1} t) \Phi_{1})
\end{equation}
The quantity $|R_{\Phi_{1}}(t)|^{2}$ is the
corresponding \emph{return probability}.

The following is a well-known connection between the positivity of the Hamiltonian and the rate of decay
(see, e.g., (\cite{KBSinha}, Lemma 5, pg. 628)). We recall that $\epsilon(t)$ means a scalar which tends 
to zero as $t \to \infty$.

\begin{theorem}
\label{th:1.2}
If, for all $t \ge 0$,
\begin{equation}
\label{(2.10.1)}
R_{\Phi_{1}}(t) = \exp(-i\lambda_{0} t) \exp(-\frac{\Gamma t}{2}) + c(t)
\end{equation}
for some $\lambda_{0} \in \mathbf{R}$ and $\Gamma > 0$, then
$$
c(t) \not\equiv 0
$$
and
$$
c(t) = \epsilon(t) \mbox{ but it is not } O(\exp(-at)) \mbox{ for any } a>0
$$
as $t \to \infty$.
\end{theorem}

We have
\begin{definition}
\label{def:1.1}
When \eqref{(2.10.1)} holds, $\lambda_{0}$ is called the \emph{level shift} and $\Gamma$ is called the
\emph{half-width} of the state $\Phi_{1}$.
\end{definition}

If we take as the (unstable or decaying by \eqref{(2.10.1)}) initial state the Weisskopf-Wigner state $\Phi_{1}$, 
we may call $\Phi_{2}(h)$ in \eqref{(1.16)} 
the ``decay products''. We refer to the version \eqref{(1.5.1)}-\eqref{(1.5.3)}. 

\begin{definition}
\label{def:1.2}
We say that there is \emph{regeneration of the unstable state from the decay products} (\cite{FGR}) if, for some
$t>0$ and some $h \in L^{2}(\mathbf{R}^{3})$,
\begin{equation}
\label{(2.10.2)}
(\exp(-itH)\Phi_{1}, \Phi_{2}(h)) \ne 0  
\end{equation}
\end{definition}

\begin{proposition}
\label{prop:1.3}
In the present model, there exists regeneration of the unstable state $\Phi_{1}$ from the decay products
according to Definition ~\ref{def:1.2} as long as $\beta \ne 0$. 
\end{proposition}

\begin{proof}

Assume \eqref{(2.10.2)} does not hold. Then, for all $t>0$, and for all $h \in L^{2}(\mathbf{R}^{3})$,
\begin{equation}
\label{(2.10.3)}
(\exp(-itH) \Phi_{1}, \Phi_{2}(h)) = 0 \mbox{ for all } t\ge 0 
\end{equation}
The right derivative of the l.h.s. of \eqref{(2.10.3)} at $t=0$ equals, however,
$$
(H \Phi_{1}, \Phi_{2}(h)) \ne 0 \mbox{ if } \beta \ne 0 
$$
due to the term $\beta \sigma_{-} \otimes a^{\dag}(g)$ in $H_{I}$ in \eqref{(1.5.3)}, if we choose $h=g$.
This contradicts \eqref{(2.10.3)}.

\end{proof}

\begin{remark}
\label{rem:1.1}
For $t$ sufficiently large, the term on the l.h.s. of \eqref{(2.10.2)} must become arbitrarily close to one, due to the last statement
of Theorem ~\ref{th:2.1}, for some given $h$ which may, however, depend on $t$. Since $(\exp(-itH_{0}) \Phi_{1}, \Phi_{2}(h)) = 0$ for all
$h \in L^{2}(\mathbf{R}^{3})$, this means that the interaction $H_{I}$ does not vanish asymptotically in time,
as happens in potential theory for short-range potentials. The expectation value of the free evolution on a decay-product state,
\begin{eqnarray*}
(\Phi_{2}(h), \exp(-itH_{0})\Phi_{2}(h))= \\
= \langle h, \exp(-it\omega(k)) h \rangle = O(t^{-2})\\
\mbox{ where } \langle h, h \rangle \\
\equiv \int \frac{d^{3}k}{2\omega(k)}|h(k)|^{2}
\end{eqnarray*}
the latter being the relativistic scalar product for the photon wave-functions: this corresponds to the correction term 
found in (\cite{Nussenzveig}, \cite{DaNu}, \cite{Dav}), and there claimed to be a consequence of causality. If the usual scalar
product is used, one obtains $O(t^{-3})$ instead. In both cases, it does not agree with the correction term $O(t^{-1})$
found in the forthcoming theorem 2.1. The asymptotic behavior of the return probability amplitude differs, therefore, 
qualitatively from that found in potential theory, where it is indeed due to the free evolution, i.e., $O(t^{-3/2})$ 
in three dimensions, whenever the potential falls off at least as fast as $|\vec{x}|^{-1-\epsilon}$, for some $\epsilon > 0$,
as $|\vec{x}| \to \infty$, i.e., faster than Coulomb, see \cite{RSIII}. Summarizing: regeneration of the unstable state 
from the decay products explains the fact that the interaction does not vanish for large times, which, on the other hand,
implies that the correction term $c(t)$ in theorem 1.2 is not due to the spreading of free photon wave packets, as is the
case in potential theory \cite{FGR}. This fact reflects the field-theoretic nature of the model. 

\end{remark}

Theorem ~\ref{th:1.2} lies at the root of the connection between the rate of of decay and positivity of the Hamiltonian.
Another important approach to this connection, also believed to be quite general, but which will only be established
within the present model in our main result in section 2, proceeds by comparing $g_{\Phi_{1}}$ in \eqref{(1.24)} with
the Lorentzian or Breit-Wigner function \eqref{(0.2)}.
 
\section{The method of decay without analyticity: the correction $c(t)$ to the Lorentzian distribution}

\subsection{The method of decay without analyticity and the time arrow problem}
\label{sec:2.1}

In this section we investigate the validity of \eqref{(2.10.1)}. We thereby avoid the use of complex energies 
and frequencies, which are associated
to the analytic continuation of the resolvent ((b.) of theorem ~\ref{th:1.1}) to ''unphysical'' Riemann sheets.
We describe this procedure by the shorthand ''the method of decay without analyticity'', which should not be 
confused with the wish to avoid any particular method of treating the problem of resonances. 

As remarked by Nussenzveig \cite{Nussenzveig}, 
the pathologies associated to the use of ''complex eigenfrequencies'' 
$\omega_{n} = \omega^{'}_{n} - i \gamma_{n}$, with $\omega^{'}_{n}$ real and $\gamma_{n}$ positive, appeared already 
in J.J.Thomson's treatment \cite{JJT} of the free modes of oscillation of the electromagnetic
field around a perfectly conducting sphere. Although $\exp(-i\omega_{n}t)= \exp(-i\omega^{'}_{n}t)\exp(-\gamma_{n}t)$
decays exponentially as $t \to \infty$, as expected from radiation damping, the corresponding radial behavior of 
free outgoing electromagnetic waves is of the form $\exp[-i\omega_{n}(t-r/c)]$, which blows up exponentially as
$r \to \infty$ (``exponential catastrophe''). A similar behavior occurs in quantum theory, associated to the 
so-called Gamow vectors (see, e.g., \cite{MWB}, section 5). Such behavior imposes the use of a space-cutoff in the 
Green functions, showing that the - a priori not physically motivated - concept of complex energies and frequencies 
is delicate, and it would be conceptually of great advantage to avoid them. We attempt to do so in this paper, 
following \cite{King}, who initiated this method in 1991.

In his paper, King \cite{King} assumed everywhere that $t \ge 0$, without mentioning it explicitly.
It happens, however, that the decay of unstable systems - atoms or particles - presents a prototypical example
of the existence of a time arrow: choosing an initial time, the decay has precisely the same behavior whatever
time direction is chosen. The problem of the arrow of time is: is there an \emph{objective} way to distinguish 
a ''future'' direction, in agreement with our general psychological perception that ''time passes''?

In \cite{Nussenzveig} it is proposed that the solution of the above-mentioned ''exponential catastrophe''
lies in the fact that the decay should be necessarily treated together
with the preparation of the state, which must have cost a finite amount of energy and have occurred at some 
finite time in the past. Our method avoids, however, the use of complex energies, and we therefore do not
find any ''exponential catastrophe''. We retain, however, Nussenzveig's suggestion as a natural and physically
compelling explanation of the assymetry between past and future, i.e., of the arrow of time, which has
been proposed in thermodynamics \cite{therm2} (see also \cite{oneortwo} for a pedagogic discussion). A similar
point of view has also been set forth by Peierls in a beautiful discussion (section 3.8, p. 73 of \cite{Peierls}).
Since the discussion is essentially identical to the one in \cite{therm2}, taking the ground state \eqref{(1.12)}
as initial state, and noting that the return probability is invariant under time-reversal by the self-adjointness
of $H_{1}$. We therefore omit it, remarking that a ``time-arrow theorem'' may be proved as a result.

\subsection{Decay without analyticity: the correction $c(t)$ to exponential decay. The main theorem and its proof}

We refer to \eqref{(1.7)} and \eqref{(1.8.1)}, \eqref{(1.8.2)} and define the functions $G$ and $F$, which will play a key role
in the following:

\begin{equation}
\label{(3.1)}
G(\lambda) \equiv g(\lambda)^{2}
\end{equation}

\begin{equation}
\label{(3.2)}
F(\lambda) \equiv vp \int_{\delta}^{\infty} \frac{G(k)}{k-\lambda} dk\\
\mbox{ for all } \delta>0 \mbox{ and } 0< \lambda< \infty
\end{equation}
where $vp$ denotes the Cauchy principal value (\cite{BB}, chapter 3.2, pg. 33). Note that for $\lambda = 0$ the principal value in \eqref{(3.2)} is not defined, but we add to \eqref{(3.2)} 
\begin{equation}
\label{(3.3)}
F(0) \equiv \lim_{\lambda \downarrow 0} F(\lambda) = \int_{0}^{\infty}\frac{G(k)}{k} dk
\end{equation}
\eqref{(3.3)} is proved in appendix A. 
By \eqref{(1.7)} and \eqref{(3.1)}, it follows that $G$ satisfies:
\begin{equation}
\label{(3.4.1)}
\sup_{k \in [0,\infty)} |G^{'}(k) (1+ k^{2})| < \infty
\end{equation}
 
The following Sokhotski-Plemelj formula (\cite{BB}, chapter 3.3, page 37) will be used:

\begin{equation}
\label{(3.4.2)}
\lim_{\epsilon \to 0} \frac{1}{x \pm i\epsilon} = \mp i\pi \delta + vp \frac{1}{x} \mbox{ in } {\cal D}^{'}(\mathbf{R})
\end{equation}
 From the proof of \eqref{(3.4.2)}, e.g., in \cite{BB}, loc.cit., it is immediately apparent that \eqref{(3.4.2)} holds as a functional on test-functions $G$ which need not belong to the Schwartz space ${\cal D}(\mathbf{R})$ but need only satisfy \eqref{(3.4.1)}. Using this fact, we obtain from \eqref{(1.21.2)}, \eqref{(1.24)}, \eqref{(3.1)} and \eqref{(3.2)} the equation
\begin{align}
\begin{split}
\label{(3.4.3)}
& \mbox{ For all } \lambda >0\\ 
& g_{\Phi_{1}}(\lambda)= \frac{d\mu_{\Phi_{1}}(\lambda)}{d\lambda} \\
& = \frac{\beta^{2}G(\lambda)}{(E_{0}-\lambda-\beta^{2}F(\lambda))^{2}+(\pi \beta^{2}G(\lambda))^{2}}
\end{split}
\end{align}  

From \eqref{(2.9)}, \eqref{(1.23.1)} and \eqref{(1.23.2)}, we obtain  

\begin{equation}
\label{(3.4.4)}
R_{\Phi_{1}}(t) = \int_{0}^{\infty} g_{\Phi_{1}}(\lambda) \exp(-i\lambda t) d\lambda
\end{equation}
\eqref{(3.4.3)}, properties a.) and b.) of $F$ in Appendix A , \eqref{(3.1)}, \eqref{(1.7)}, and \eqref{(1.8.1)} imply that 
$g_{\Phi_{1}}(\lambda)$ in \eqref{(3.4.4)} is uniformly bounded in $\lambda$ near $\lambda=0$ and of decay $O(\lambda^{-7})$ for
large $\lambda$, so that the integral on the r.h.s. of \eqref{(3.4.4)} is well defined. We may now state our main theorem:

\begin{theorem}
\label{th:2.1}

There exists a constant $b>0$ such that, if 
\begin{equation}
\label{(3.4.5)} 
\beta < b
\end{equation}
then 
\begin{equation}
\label{(3.4.6)}
R_{\Phi_{1}}(t) = (1+O(\beta^{2}))\exp(-i\lambda_{0}t)\exp(-\frac{\Gamma t}{2})+c(t)
\end{equation}
with the level shift $\lambda_{0}$ given by the unique solution in a sufficiently small neighbourhood of $E_{0}>0$ of the
equation
\begin{equation}
\label{(3.5)}
E_{0} -\lambda_{0} -\beta^{2} F(\lambda_{0})=0
\end{equation}
and the half-width $\Gamma$ is given by
\begin{equation}
\label{(3.6)}
\Gamma = 2\pi \beta^{2} G(E_{0}) 
\end{equation}
Furthermore, in \eqref{(3.4.6)}, $c(t)$ is given by
\begin{equation}
\label{(3.7)}
c(t) = c_{1}(t)+c_{2}(t)
\end{equation}
where
\begin{equation}
\label{(3.8.1)}
\lim_{t \to \infty} tc_{1}(t) = \frac{\beta^{2}d}{E_{0}}
\end{equation}
for some constant $d>0$ independent of $\beta$ and
\begin{equation}
\label{(3.8.2)}
|c_{2}(t)| \le \frac{c \beta^{2}}{t}
\end{equation}
for all $t>0$ and $c>0$ independent of $t$.

\end{theorem}

\begin{proof}

As in \cite{King}, the strategy of the proof will be to approximate $g_{\Phi_{1}}$, given by \eqref{(3.4.3)}, by a Lorentzian
(or Breit-Wigner) function: this will yield \eqref{(3.4.6)}, with \eqref{(3.7)} and $c_{2}=0$, and $c_{1}$ satisfying \eqref{(3.8.1)}.
An estimate of the remainder provides then \eqref{(3.7)}, with $c_{2}$ satisfying \eqref{(3.8.2)}.

We expand, as in \cite{King}, \eqref{(3.4.3)} around $\lambda=\lambda_{0}$ (the solution of \eqref{(3.5)} under assumption \eqref{(3.4.5)},
which exists by the implicit function theorem under our assumptions on $G$ and $F$, in particular the continuous 
differentiability of $F$ in a neighborhood of $E_{0}$) to second order. Define
\begin{equation}
\label{(3.9)}
\kappa \equiv -1-\beta^{2}F^{'}(\lambda_{0})-i\pi\beta^{2}G^{'}(\lambda_{0})
\end{equation}
where the prime indicates differentiation. Then
\begin{align}
\begin{split}
\label{(3.10)}
& \alpha(\lambda) \equiv E_{0}-\lambda-\beta^{2}F(\lambda)-i\pi\beta^{2}G(\lambda)\\
& = \kappa(\lambda-\lambda_{0})-i\pi\beta^{2}G(\lambda_{0})+w(\lambda)
\end{split}
\end{align}
where the remainder $w(\lambda)$ in \eqref{(3.10)} is equal to
\begin{align}
\begin{split}
\label{(3.11)}
w(\lambda) & =-\beta^{2}[F(\lambda)-F(\lambda_{0})-F^{'}(\lambda_{0})(\lambda-\lambda_{0})]\\
           &-i\pi\beta^{2}[G(\lambda)-G(\lambda_{0})-G^{'}(\lambda_{0})(\lambda-\lambda_{0})]
\end{split}
\end{align}
From \eqref{(3.9)},

\begin{align}
\begin{split}
\label{(3.12)}
(\kappa)^{-1} & = (-1-\beta^{2}F^{'}(\lambda_{0})-i\pi\beta^{2}G^{'}(\lambda_{0}))^{-1}\\
             & = -(1+\beta^{2}A)^{-1}
\end{split}
\end{align}
where
\begin{equation}
\label{(3.13)}
A \equiv F^{'}(\lambda_{0})+i\pi G^{'}(\lambda_{0})
\end{equation}

From \eqref{(3.12)}
\begin{equation}
\label{(3.14)}
(\kappa)^{-1} = -[1-\beta^{2}A+B(\beta^{2}A)^{2}]
\end{equation}
where
\begin{equation}
\label{(3.15)}
|B| \le 2
\end{equation}
if
\begin{equation}
\label{(3.16)}
\beta^{2}\sqrt{[F^{'}(\lambda_{0})^{2}+\pi^{2}G^{'}(\lambda_{0})^{2}]} < \frac{1}{2}
\end{equation} 

Thus, a Lorentzian (or Breit-Wigner) approximation to $g_{\Phi_{1}}$, given by \eqref{(1.23.2)} or \eqref{(1.24)}, is
\begin{equation}
\label{(3.17)}
L(\lambda) \equiv \frac{1}{\pi}\Im\frac{1}{\kappa}(\lambda-\lambda_{0}-i\pi \beta^{2}\kappa^{-1}G(\lambda_{0}))^{-1}
\end{equation}
where
\begin{equation}
\label{(3.18)}
\kappa^{-1}G(\lambda_{0})=-G(\lambda_{0})+O(\beta^{2})
\end{equation}
by \eqref{(3.14)}-\eqref{(3.16)}. By \eqref{(3.17)} and \eqref{(3.18)}, the point
\begin{equation}
\label{(3.19)}
\bar{\lambda} \equiv \lambda_{0}+i\pi \beta^{2}\kappa^{-1}G(\lambda_{0})=\lambda_{0}-i\pi \beta^{2}G(\lambda_{0})+O(\beta^{4})
\end{equation}
lies on the lower half of the complex plane. Accordingly, we write
\begin{equation}
\label{(3.20.1)}
R_{\Phi_{1}}(t) = I_{L}(t) + D_{L}(t)
\end{equation}
where
\begin{equation}
\label{(3.20.2)}
I_{L}(t) \equiv \int_{0}^{\infty} \exp(-it\lambda)L(\lambda)d\lambda
\end{equation}
and
\begin{equation}
\label{(3.20.3)}
D_{L}(t)\equiv \int_{0}^{\infty} \exp(-it\lambda)(g_{\Phi_{1}}(\lambda)-L(\lambda))d\lambda
\end{equation}
We apply Cauchy's theorem to the complex integral of
\begin{equation}
\label{(3.21)}
f(z) \equiv \exp(-itz) L(z)
\end{equation}
along the clockwise circuit $\Gamma \equiv C_{1}\cup C_{2} \cup (-C_{3})$, where $C_{1} \equiv \{iy;-R \le y \le 0\}$, $C_{2}=[0,R]$,
and $C_{3}= \{\exp(i\theta);-\frac{\pi}{2} \le \theta \le 0\}$, and let $R \to \infty$, avoiding the pole $\bar{\lambda}$. The contribution
of $C_{3}$ tends to zero due to the term $\exp(-itz)$ in \eqref{(3.20.2)} (recall that $t>0$). We now estimate that of $C_{1}$, writing
first
\begin{align}
\begin{split}
\label{(3.22)}
&L(\lambda)= \frac{1}{\pi} \Im \frac{1}{\kappa(\lambda-\lambda_{0})-i\pi \beta^{2}G(\lambda_{0})}=\\  
&= \frac{1}{2\pi i}(\frac{1}{\kappa(\lambda-\lambda_{0})-i\pi \beta^{2}G(\lambda_{0})}-\\
&-\frac{1}{\kappa(\lambda-\lambda_{0})+i\pi \beta^{2}G(\lambda_{0})})
\end{split}
\end{align}
Therefore, by \eqref{(3.20.2)},
\begin{equation}
\label{(3.23.1)}
I_{L}(t) = -2\pi i res(\bar{\lambda})-\frac{\beta^{2}G(\lambda_{0})}{t}\int_{0}^{\infty}dy \exp(-y)f(t,y)
\end{equation}
where
\begin{equation}
\label{(3.23.2)} 
f(t,y)\equiv [\kappa(\frac{-iy}{t}-\lambda_{0})-i\pi \beta^{2}G(\lambda_{0})]^{-1}[\kappa(\frac{-iy}{t}-\lambda_{0})+i\pi \beta^{2}G(\lambda_{0})]^{-1}
\end{equation}
By \eqref{(3.19)} and \eqref{(3.21)},
\begin{equation}
\label{3.24)}
res(\bar{\lambda}) = \exp(-it\lambda_{0})\exp(-\pi \beta^{2}G(\lambda_{0})t)[1+O(\beta^{2})]
\end{equation}
We have
\begin{eqnarray*}
|\kappa(\frac{-iy}{t}-\lambda_{0})-i\pi \beta^{2}G(\lambda_{0})| \\
\ge |\kappa| |\frac{-iy}{t}-\lambda_{0}| - \pi \beta^{2}G(\lambda_{0})\\
\ge ((1-O(\beta^{2}))\lambda_{0} -\pi \beta^{2}G(\lambda_{0})) \ge \lambda_{0} - O(\beta^{2})
\end{eqnarray*}
and similarly for the other denominator in \eqref{(3.23.2)}, by \eqref{(3.12)}-\eqref{(3.16)}. Hence, by \eqref{(3.23.2)}
\begin{equation}
\label{(3.25)}
|f(t,y)| \le (\lambda_{0}-O(\beta^{2}))^{-2}
\end{equation}
By \eqref{(3.23.1)}, \eqref{(3.23.2)}, \eqref{(3.25)} and the Lebesgue dominated convergence theorem, we obtain the \eqref{(3.8.1)}-part of
\eqref{(3.4.6)} of Theorem ~\ref{th:2.1}.

\qquad

We now prove that $D_{L}(t)$, defined by \eqref{(3.20.3)}, satisfies the bound
\begin{equation}
\label{(3.26.1)}
|D_{L}(t)| \le \frac {c \beta^{2}}{t} \mbox{ for all } t>0
\end{equation}
where $c$ is a constant, independent of $\beta$ and $t$. Together with \eqref{(3.20.1)}, this proves \eqref{(3.8.2)}. By definition
\eqref{(3.20.3)}, \eqref{(3.4.3)}, \eqref{(3.10)} and \eqref{(3.11)}, we find
\begin{equation}
\label{(3.27.1)}
D_{L}(t) = D_{L}^{(1)}(t) - \overline{D_{L}^{(1)}(-t)}
\end{equation}
where
\begin{equation}
\label{(3.27.2)}
D_{L}^{(1)}(t) = \frac{1}{2\pi i}\int_{0}^{\infty} \exp(-it\lambda)\frac{w(\lambda)}{\beta(\lambda)\alpha(\lambda)}d\lambda
\end{equation}
In \eqref{(3.27.1)}, the bar denotes complex conjugate. In \eqref{(3.27.2)}, $\alpha(\lambda)$ is given by \eqref{(3.10)}
and 
\begin{equation}
\label{(3.27.3)}
\beta(\lambda) \equiv \kappa(\lambda - \lambda_{0}) -i\pi \beta^{2} G(\lambda_{0})
\end{equation}
By \eqref{(3.27.1)} and \eqref{(3.27.2)}, in order to prove \eqref{(3.8.2)}, it suffices to prove
\begin{equation}
\label{(3.26.2)}
|D_{L}^{1}(t)| \le \frac {c \beta^{2}}{t} \mbox{ for all } t>0
\end{equation}
The proof of \eqref{(3.26.2)} is given in appendix A.

\end{proof}

\begin{remark}
\label{rem:2.2}

Instead of the splitting \eqref{(3.20.1)}, King \cite{King} defines (in our notation)
\begin{equation}
\label{(3.21.5)}
I_{L}(t) \equiv \int_{-\infty}^{\infty} \exp(-it\lambda) L(\lambda) d\lambda
\end{equation}
He thereby adds to $R_{\Phi_{1}}(t)$ a term
$$
I_{L}^{'}(t) \equiv \int_{-\infty}^{0} \exp(-it\lambda) L(\lambda) d\lambda
$$
By \eqref{(3.9)} and \eqref{(3.17)}, $L(\lambda)$ is $O(1)$ and not $O(\beta^{2})$. In our view, it happens that it is just the fact that $I_{L}(t)$ is given by \eqref{(3.20.2)} - and not \eqref{(3.21.5)} - which is responsible for the universal term
$c_{1}(t) = O(\beta^{2}\frac{1}{t})$ in theorem 2.1. The rest of the proof of theorem 2.1 is devoted to establishing that the (non-universal) correction to the Lorentzian term does not alter this conclusion qualitatively, as demonstrated by \eqref{(3.4.6)}, \eqref{(3.7)}, \eqref{(3.8.1)} and \eqref{(3.8.2)} of that theorem.

\end{remark}

\section{Sojourn time, its physical interpretation and a time-energy uncertainty relation}

Since $\Gamma$ is the most fundamental physical quantity characterizing decay, it would both more elegant and conceptually more advantageous to characterize it by a global quantity - i.e., not relying on pointwise estimates in the time variable, such as \eqref{(2.10.1)}.                                

 This subject has a very long history, well summarized in the introduction to the article of Gislason, Sabelli and Wood \cite{GSW}, with various important references: it is known under the general heading of ''time-energy uncertainty relation''. More recent reviews of the topic, which also added significant new results, are the articles by Brunetti and Fredenhagen \cite{BrF2} and Pfeifer and Fr\"{o}hlich \cite{PFro}, as well as the book \cite{Bus}, to which we also refer for additional references.

An initial relevant remark is that the early version of the time-energy uncertainty relation, stating that, if the energy of a system is measured during a time $\Delta t$, the corresponding uncertainty $\Delta E$ in the energy variable $E$ must satisfy $\Delta E \Delta t \ge \frac{1}{2}\hbar$, is physically untenable, because, as reviewed in the introduction to \cite{GSW}, it seems generally accepted that the energy of a system can be measured with arbitrary precision and speed. This was first pointed out by Aharonov and Bohm \cite{ABohm}. The point we wish to make is that the very designation ''time-energy uncertainty relation'' is inadequate, because the quantity multiplying $\Delta E$ in the would-be inequality is of entirely different nature from ''$\Delta t$''. Our results in this section bring a new light on this matter.

We assume a slightly more general setting than in previous sections. Let $H$ be a self-adjoint operator on a Hilbert space ${\cal H}$, and, for $\Psi \in {\cal H}$, define
\begin{equation}
\label{(4.1)}
R_{\Psi}(t) = (\Psi, \exp(-itH) \Psi)
\end{equation}
This is just the return probability amplitude for the vector $\Psi$, given by \eqref{(2.9)}. For some $\Psi_{0} \in {\cal H}$, assume that
\begin{equation}
\label{(4.2)}
R_{\Psi_{0}} \in L^{2}(-\infty,\infty)
\end{equation}
and define the \emph{sojourn time of the system in the state $\Psi_{0}$} (\cite{KBSinha1}, \cite{BDFS}) by
\begin{equation}
\label{(4.3)}
\tau_{H}(\Psi_{0}) \equiv \int_{0}^{\infty}|R_{\Psi_{0}}(t)|^{2}dt
\end{equation}
By a theorem of Sinha \cite{KBSinha1}, \eqref{(4.2)} requires that $H$ have purely absolutely continuous (a.c.) spectrum. A lower bound to the sojourn time is given by the rigorous version of the Gislason-Sabelli-Wood time-energy uncertainty relation proved in (\cite{MWB}, Theorem 3.17, page 81):

\begin{theorem}[rigorous version of the theorem of Gislason-Sabelli-Wood \cite{GSW}]
\label{th:3.1}
Let \eqref{(4.2)} hold and
\begin{equation}
\label{(4.4)}
\Psi_{0} \in D(H) \mbox{ i.e., } ||H\Psi_{0}||<\infty
\end{equation}
Then
\begin{equation}
\label{(4.5)}
I_{H}(\Psi_{0}) \equiv \tau_{H}(\Psi_{0})\Delta E \ge \frac{3\pi \sqrt(5)}{25}
\end{equation}
where
\begin{equation}
\label{4.6)}
(\Delta E)^{2} \equiv (\Psi_{0},H^{2}\Psi_{0})-(\Psi_{0},H \Psi_{0})^{2}
\end{equation}
is the energy variance (uncertainty) in the state $\Psi_{0}$.
\end{theorem}

This theorem has been applied to estimate the half-widths of negative ion resonances in \cite{deRoseGS}.

In order to assess the physical meaning of $\tau_{H}(\Psi_{0})$, let, following \cite{GSW}, 

\begin{equation}
\label{(4.7.1)}
Q(t) \equiv |R_{\Psi_{0}}(t)|^{2}
\end{equation}
denote the (quantum) probability that the system has \emph{not} decayed up to the time $t$. Then the quantity
$$
Q(t) - Q(t+\Delta t) = -Q^{'}(t)\Delta t + o(\Delta t)
$$
equals the quantum probability that the system has decayed in the interval $[t,t+\Delta t)$, and thus the \emph{average lifetime} $\tau$
of the decaying state is
\begin{equation}
\label{(4.7.2)}
\tau= -\int_{0}^{\infty}dt t Q^{'}(t)= [tQ(t)]_{0}^{\infty}+ \int_{0}^{\infty}dt Q(t)=\tau_{H}(\Psi_{0})
\end{equation}
as long as 
\begin{equation}
\label{(4.8)}
\lim_{t \to \infty} tQ(t) = 0
\end{equation}

Our main result in this section is the following theorem:

\begin{theorem}
\label{th:3.2}
For model \eqref{(1.15)}, \eqref{(4.2)}, as well as \eqref{(4.4)}, are true, if $\Psi_{0}=\Phi_{1}$, the Weisskopf-Wigner state. Moreover:
\begin{itemize}
\item [$a.)$] 
\begin{equation}
\label{(4.9)}
\Delta E \ge 0.843 \Gamma
\end{equation}
\item [$b.$]
Equation \eqref{(4.8)} holds, and therefore the time of sojourn has the interpretation of an average lifetime.
\end{itemize}
\end{theorem}

\begin{proof}

\eqref{(4.2)} follows directly from Theorem 2.1. By the spectral theorem,
\begin{equation}
\label{(4.10)}
||H \Psi_{0}||^{2} = \int_{-\infty}^{\infty} d\lambda \lambda^{2} g_{\Phi_{1}}(\lambda)
\end{equation} 
In \eqref{(3.4.3)}, by \eqref{(1.7)}, \eqref{(1.8.1)}, \eqref{(3.1)}, the numerator $G(\lambda)$ decays as $|\lambda|^{-7}$ for large $|\lambda|$, and
$$
|\frac{\lambda^{2}}{(E_{0}-\lambda-\beta^{2} F(\lambda))^{2}+(\pi \beta^{2} G(\lambda))^{2}}| \le c
$$
where the constant $c$ independs of $\lambda$ and the other parameters, by property $a.)$ of $F(\lambda)$ proved in Appendix A. Thus,
$$
\int_{-\infty}^{\infty} d\lambda \lambda^{2} g_{\Phi_{1}}(\lambda) < \infty
$$
which, together with \eqref{(4.10)}, proves \eqref{(4.4)}.

Further estimate of $\tau_{H}(\Phi_{1})$ depends on a suitable splitting of the time interval into three parts, corresponding to ``small'' $t \le t_{\epsilon}$, ``intermediate'' $t_{\epsilon} \le t \le t_{0}$, and ``large'' $t \ge t_{0}$, which we omit. The latter part concerns the correction $c(t)$ in \eqref{(3.4.6)} and yields the term
$$
\int_{t_{0}}^{\infty} dt \frac{|c|^{2} \beta^{4}}{E_{0}^{2}t^{2}} = \frac{|c|^{2}\beta^{4}}{E_{0}^{2}t_{0}}
$$
for $|c|$ of order one, this term is of order $\alpha^{8} \approx 10^{-16}$.
We further choose $t_{0}$ such that
\begin{equation}
\label{(4.17)}
\exp(-\frac{t}{2\tau}) \ge \frac{|c|\beta^{2}}{E_{0}t} \mbox{ if } t_{\epsilon}\le t \le t_{0}
\end{equation}
With these choices, it follows that
\begin{equation}
\label{(4.18)}
|\tau_{H}(\Phi_{1}) - \frac{1}{\Gamma}| \le 10^{-4} \frac{1}{\Gamma}
\end{equation}
By Theorem 2.1 and \eqref{(4.7.2)}, it follows that $Q(t)=O(\frac{1}{t^{2}})$ for large $t$, so that \eqref{(4.8)} holds, and thus b.).

\end{proof}

\begin{remark}
\label{rem:3.1}

The interest of \eqref{(4.9)} is better appreciated by realizing that the method of proof of Theorem 2.1, i.e., comparison with the Lorentzian $L(\lambda)$, fails for $\Delta E$, because the r.h.s. of \eqref{(4.10)}, when $g_{\Phi_{1}}(\lambda)$ is replaced by $L(\lambda)$, is infinite.

Further, \eqref{(4.18)} shows that the sojourn time equals indeed, to a very good approximation, the inverse half-width of the state. This is due to the apparently general fact that, both in atomic and particle physics, the Lorentzian (Breit-Wigner) approximation is excellent - as seen from \eqref{(4.17)} and the fact that, after 48 lifetimes, the atom ''has decayed for all practical purposes'', as remarked by Nussenzveig in \cite{Nussenzveig}.

\end{remark}

\begin{remark}
\label{rem:3.2}

In order that the level shift $\lambda_{0}-E_{0}$ may be measured with great precision, as is the case of the Lamb shift, it is crucial that it is of lower order than the width. It seems remarkable that this is so even in this simple model, where $\lambda_{0}-E_{0} = O(\beta^{2})= O(\alpha^{3})$, and 
$\Gamma = 2\pi \beta^{2} G(E_{0}) \approx 2\pi \beta^{2}E_{0} = O(\alpha^{3})\alpha= O(\alpha^{4})$, since $E_{0}=O(\alpha)$.

\end{remark}

\section{Conclusion}

In Theorem 2.1 we proved that positivity of the Hamiltonian $H$ implies \eqref{(3.4.6)}, with $c(t) = O(\beta^{2}\frac{1}{t})$ for sufficiently large positive times and sufficiently small coupling constant $\beta$. This correction  is universal and improves on some results
of \cite{King}. The remaining, non-universal, part of the correction is also shown to be of the same qualitative type. The method consists in approximating the matrix element of the resolvent operator in the W.W. state by a Lorentzian distribution. No use is made of complex energies associated 
to analytic continuations of the resolvent operator to ''unphysical'' Riemann sheets.

The above-mentioned correction, although very small and negligible for the computation of the half-width $\frac{1}{\Gamma}$ (theorem 3.2), plays nevertheless a basic conceptual role. It is due to the regeneration of the decaying state from the decay products, a virtual process which is of the same nature of the tunneling which plays a crucial role in the Gamow theory of alpha decay (\cite{Gamow}, \cite{BrHa}) but, unlike the latter, is characteristic of a quantum field theory (see remark 1.1).

Due to Sinha \cite{KBSinha1} and Lavine \cite{Lav1} is the concept of sojourn time $\tau_{H}(\Psi)$ given by \eqref{(4.3)}. As a functional over a particular set of elements $\Psi$ of the Hilbert space ${\cal H}$, on which the self-adjoint operator is defined (e.g., in potential theory, the set of Kato-smooth vectors, see \cite{RSIII} and \cite{Lav1}, the problem was posed by the late Pierre Duclos (see also \cite{BDFS}) of obtaining lower bounds to $\tau_{H}(\Psi)$, motivated by the expectation that, near resonances, $\tau_{H}$ assumes very large values; one lower bound was given by Lavine's form of the time-energy uncertainty relation \cite{Lav1} (see also \cite{Asch} for a new version and an improvement of Lavine's bounds), another by the rigorous form of the Gislason-Sabelli-Wood time-energy uncertainty relation, Theorem 3.1. The application to the present model (Theorem 3.2) shows that the sojourn time is the physically most natural concept describing decay, because it coincides with the average lifetime of the state, a standard concept in quantum probability.

In spite of its simplicity, the present model has some surprisingly realistic features (see, e.g., remark ~\ref{rem:3.2}). Its most unrealistic aspect is, of course, the lack of vacuum polarization, which allows us to work in Fock space and yields an unphysical conservation law, which is, however, responsible for the relatively easy estimates of the time evolution, viz., of the return probability amplitude of the Weisskopf-Wigner state. In fact, we know of no other model in which a closed form exists for the expectation value of the resolvent on a particular state, which simulates an explicit ``pole term'' in the lower half-plane as a consequence of the interaction - a fact we find remarkable.

When the ``counterrotating'' term
$$
H_{I}^{'} = \beta[\sigma_{+}\otimes a^{\dag}(g)+\sigma_{-}\otimes a(g)]
$$
is added to $H$, the above picture no longer holds, but a perturbative treatment (\cite{Dav}, see also \cite{DaNu}) is available: the final results for the Lamb shift, as well as for the line shape, are in good agreement with experiment.

Our new result may be very simply stated. The presence of a term simulating a ``pole term'' in the matrix element of the resolvent in the W.W. state allows to use Cauchy's theorem, as in \cite{King}. We do use Cauchy's theorem, but point out that, upon use of a convenient contour which takes semi-boundedness of $H$ into account (in contrast to \cite{King}), the main part of the correction to the Lorentzian arises already. This correction turns out to be of qualitatively different nature as the analogous one in potential theory, which arises from the spreading of the wave-packet, as discussed in remark 1.1.

Acknowledgements

We should like to thank the first referee for his encouraging remarks and corrections. We are also deeply indebted to the second one for important remarks and corrections, as well as a very thorough reading of the painful details of this article.

\section{Appendix A - completion of the proof of Theorem 2.1}

In this appendix we prove that (74) of Theorem 2.1 holds. Together with (71), this proves (70), and thereby completes the proof of Theorem 2.1. 

We first write (72) as the limit, as $\delta \downarrow 0$, of the corresponding integral from $\delta >0$ to $\infty$. By integration by parts on the latter, we find
\begin{eqnarray*}
D_{L}^{1}(t) = \lim_{\delta \downarrow 0}[-\frac{w(\delta)}{it\alpha(\delta)\beta(\delta)} +\\
+ \frac{\int_{\delta}^{\infty}d\lambda \exp(-it\lambda)\frac{d}{d\lambda}(\frac{w(\lambda)}{\alpha(\lambda)\beta(\lambda)})}{it}]
\end{eqnarray*}
$$\eqno{(A.1)}$$
where, for $\lambda>0$, $\alpha(\lambda)$ and $\beta(\lambda)$ are given by (51) and (73) of the main text, but we repeat them here for clarity:
$$
\alpha(\lambda) \equiv E_{0}-\lambda-\beta^{2}F(\lambda)-i\pi\beta^{2}G(\lambda)
\eqno{(A.2)}
$$
and
$$
\beta(\lambda) \equiv \kappa(\lambda-\lambda_{0})-i\pi \beta^{2}G(\lambda_{0})
\eqno{(A.3)}
$$
We have, the prime denoting, as usual, the first derivative,
$$
\alpha^{'}(\lambda) = -1-\beta^{2}F^{'}(\lambda)-i\pi \beta^{2}G^{'}(\lambda)
\eqno{(A.4)}
$$
and
$$
\beta^{'}(\lambda) = \kappa
\eqno{(A.5)}
$$
From (52),
$$
w^{'}(\lambda)=-\beta^{2}(F^{'}(\lambda)-F^{'}(\lambda_{0}))-i\pi \beta^{2}(G^{'}(\lambda)-G^{'}(\lambda_{0}))
\eqno{(A.6)}
$$
\begin{eqnarray*}
\frac{d}{d\lambda}(\frac{w(\lambda)}{\alpha(\lambda)\beta(\lambda)}) = \\
= \frac{w^{'}(\lambda)}{\alpha(\lambda)\beta(\lambda)} - \frac{w(\lambda)\alpha^{'}(\lambda)}{\alpha(\lambda)^{2}\beta(\lambda)}-\\
- \frac{w(\lambda)\beta^{'}(\lambda)}{\alpha(\lambda)\beta(\lambda)^{2}}
\end{eqnarray*}
$$\eqno{(A.7)}$$
From (21),(22),(36) and (37) we have
$$
G(\lambda) = \lambda (\lambda^{2}+a^{2})^{-4} \mbox{ for } \lambda \ge 0
\eqno{(A.8.1)}
$$
$$
G^{'}(\lambda) = (\lambda^{2}+a^{2})^{-4}-8\lambda^{2}(\lambda^{2}+a^{2})^{-5}
\eqno{(A.8.2)}
$$
When writing $f(0)$ in the following, for some function $f$, it will be meant the limit $\lim_{\delta \downarrow 0} f(\delta)$. The finiteneness of the resulting limits, for all the functions which follow, will result from (38), which will be proved later as part of the forthcoming property b.) of the function $F$. We have, then: 
$$
F(0) = \int_{0}^{\infty}(k^{2}+a^{2})^{-4}dk
\eqno{(A.8.3)}
$$
$$
G(0) = 0
\eqno{(A.8.4)}
$$
$$
w(0)= -\beta^{2}[F(0)-F(\lambda_{0})+\lambda_{0}F^{'}(\lambda_{0})]-i\pi \beta^{2}\lambda_{0}G^{'}(\lambda_{0})
\eqno{(A.8.5)}
$$
$$
\alpha(0) = E_{0}-\beta^{2}F(0)
\eqno{(A.8.6)}
$$
$$
\beta(0) = -\kappa \lambda_{0}-i\pi \beta^{2}G(\lambda_{0})
\eqno{(A.8.7)}
$$
The first term in (A.1) satisfies, in the limit $\delta \downarrow 0$, the bound on the r.h.s. of (75), by (A.8.5), (A.8.6) and (A.8.7). Therefore, by (A.1) and (A.7), in order to conclude the proof of (74), we need only prove that
$$
|\int_{0}^{\infty} \frac{\alpha^{'}(\lambda)}{\alpha(\lambda)^{2}\beta(\lambda)} w(\lambda) d\lambda|<\infty
\eqno{(A.9.1)}
$$
$$
|\int_{0}^{\infty} \frac{\beta^{'}(\lambda)}{\alpha(\lambda)\beta(\lambda)^{2}} w(\lambda) d\lambda| <\infty
\eqno{(A.9.2)}
$$
$$
|\int_{0}^{\infty} \frac{1}{\alpha(\lambda)\beta(\lambda)} w^{'}(\lambda) d\lambda| <\infty
\eqno{(A.9.3)}
$$
It follows from (A.2), (A.3), (A.4), (A.5), (A.8.1) and (A.8.2) and (52) that (A.9.1)-(A.9.3) hold if the two following assertions are true:
\begin{itemize}
\item [$a.)$]
For $\lambda$ sufficiently large, $F(\lambda)$ and $F^{'}(\lambda)$ are uniformly bounded in $\lambda$;
\item [$b.)$]
For $\lambda$ in a sufficiently small right-neighbourhood of zero, $F(\lambda)$ is uniformly bounded, (38) holds and
$$
F^{'}(\lambda) = -\log \lambda + D 
$$
where $0<D<\infty$ is independent of $\lambda$.
\end{itemize}
Indeed, $b.)$ implies that $\alpha^{'}$, as well as $w^{'}$, are integrable in a neighbourhood of zero, which suffice to prove integrability of
$\frac{\alpha^{'}(\lambda)}{\alpha(\lambda)^{2}\beta(\lambda)} w(\lambda)$ and of $\frac{1}{\alpha(\lambda)\beta(\lambda)}w^{'}(\lambda)$, in a 
neighbourhood of zero, which are elements in the proof of (A.9.1) and (A.9.3). Convergence at infinity of the integrals on the left hand sides of (A.9.1)-(A.9.3) is an immediate consequence of the explicit formulae for $\alpha$, $\beta$ and $w$, together with $a.)$.

In order to prove $a.)$ and $b.)$, we come back to (37), whereby, for any $\lambda>0$,
$$
F(\lambda) = \lim_{r\to 0}\int_{|k-\lambda|\ge r} \frac{G(k)}{k-\lambda}dk
$$
We write
\begin{eqnarray*}
\int_{|k-\lambda|\ge r} \frac{G(k)}{k-\lambda} =\\
=\int_{0}^{\lambda-r} \frac{G(k)}{k-\lambda}dk + \int_{\lambda+r}^{2\lambda} \frac{G(k)}{k-\lambda}+\\
+\int_{2\lambda}^{\infty} \frac{G(k)}{k-\lambda}dk
\end{eqnarray*}
but
\begin{eqnarray*}
\int_{0}^{\lambda-r} \frac{G(k)}{k-\lambda}dk + \int_{\lambda+r}^{2\lambda} \frac{G(k)}{k-\lambda} =\\
= \int_{r}^{\lambda} \frac{1}{k}[G(k+\lambda)-G(k-\lambda)]dk
\end{eqnarray*}
Write
$$
G(k+\lambda)-G(k-\lambda) = k\int_{-1}^{1}dt G^{'}(\lambda+kt)
$$
Thus,
\begin{eqnarray*}
F(\lambda) = \int_{0}^{\lambda}dk\int_{-1}^{1}dt\{[(\lambda+kt)^{2}+a^{2}]^{-4}\\
-8(\lambda+kt)^{2}[(\lambda+kt)^{2}+a^{2}]^{-5}\}+\\
+ \int_{2\lambda}^{\infty} \frac{G(k)}{k-\lambda}dk
\end{eqnarray*}$$\eqno{(A.11)}$$
We write
\begin{eqnarray*}
F(\lambda) = -7\int_{0}^{\lambda}dk\int_{-1}^{1}dt[(\lambda+kt)^{2}+a^{2}]^{-4}+\\
+8a^{2} \int_{0}^{\lambda}dk\int_{-1}^{1}dt [(\lambda+kt)^{2}+a^{2}]^{-5}+\\
+ \int_{2\lambda}^{\infty} \frac{G(k)}{k-\lambda}dk
\end{eqnarray*}
from which
\begin{eqnarray*}
F^{'}(\lambda) = -7\int_{-1}^{1}dt[\lambda^{2}(1+t)^{2}+a^{2}]^{-4} +\\
+ 8a^{2} \int_{-1}^{1} dt [\lambda^{2}(1+t)^{2}+a^{2}]^{-5}+\\
+ 28 \int_{0}^{\lambda}dk\int_{-1}^{1} [(\lambda+kt)^{2}+a^{2}]^{-5}2(\lambda+kt)\\
-40a^{2} \int_{0}^{\lambda}dk\int_{-1}^{1}dt [(\lambda+kt)^{2}+a^{2}]^{-6}2(\lambda+kt)\\
-2 \int_{2\lambda}^{\infty} \frac{G(k)}{k-\lambda}dk - \int_{2\lambda}^{\infty} \frac{G(k)}{(k-\lambda)^{2}}dk
\end{eqnarray*}$$\eqno{(A.12)}$$
By (A.11), we obtain directly $a.)$ for $F(\lambda)$, as well as the statements in b.) which concern $F(\lambda)$. Statement $b.)$ for $F^{'}(\lambda)$ follows from (A.8.1) and the last term in (A.12). Statement a.) for $F^{'}(\lambda)$ is not entirely obvious from (A.12), but we use 
$$
b(\lambda+kt) \le (\lambda+kt)^{2} +a^{2}
$$
which is true for $b$ sufficiently small, to bound the third and fourth terms in (A.12) in absolute value by
$$
\mbox{ const. } \int_{0}^{\lambda} dk((\lambda-k)^{2}+a^{2})^{-4} \mbox{ resp. const. } \int_{0}^{\lambda}dk ((\lambda-k)^{2}+a^{2})^{-5}
$$
which are trivially seen to be uniformly bounded in $\lambda$ by a change of variable. This completes the proof of (74). q.e.d.

\end{document}